\documentclass[10pt]{article}
\usepackage{graphicx,amsmath,amssymb}
\usepackage{subcaption}
\usepackage{hyperref}
\usepackage{paralist}
\usepackage{color}
\usepackage{tikz,pgfkeys}
\usepackage{tkz-graph}
\usepackage{amsthm}
\usetikzlibrary{arrows,shapes}
\usepackage{mathptmx}
\usepackage{setspace}
\usepackage{tabularx}

\newtheorem{theorem}{Theorem}[section]
\newtheorem{lemma}[theorem]{Lemma}

\newenvironment{definition}{$\;$\newline \noindent {\bf
    Definition}$\;$}{$\;$\newline}

\def\boxit#1{\vbox{\hrule\hbox{\vrule\kern4pt
  \vbox{\kern1pt#1\kern1pt}
\kern2pt\vrule}\hrule}}

\begin{document}

\title{ \bf An $O^*(1.84^k)$ Parameterized Algorithm for \\ the Multiterminal Cut Problem}

\author{
  Yixin Cao\thanks{{Institute for Computer Science and Control},
    Hungarian Academy of Sciences, Hungary, {\tt yixin@sztaki.hu}.
    Supported by the European Research Council (ERC) under the 
    grant 280152 and the Hungarian Scientific Research Fund 
    (OTKA) under the grant NK105645.}
  \and 
  Jianer Chen\thanks{Department of Computer Science and Engineering,
    Texas A\&M University, U.S.A., and School of Information Science \&
    Engineering, Central South University, P.R. China, {\tt
      chen@cse.tamu.edu}.  Supported in part by US NSF under the
    grants CCF-0830455 and CCF-0917288.}
  \and
  Jia-Hao Fan\thanks{Department of Computer Science and Engineering,
    Texas A\&M University, U.S.A., {\tt grantfan@cse.tamu.edu}.}
}

\date{}
\maketitle

\begin{abstract}
  We study the \emph{multiterminal cut} problem, which, given an
  $n$-vertex graph whose edges are integer-weighted and a set of
  terminals, asks for a partition of the vertex set such that each
  terminal is in a distinct part, and the total weight of crossing
  edges is at most $k$.  Our weapons shall be two classical results
  known for decades: \emph{maximum volume minimum ($s,t$)-cuts} by
  [Ford and Fulkerson, \emph{Flows in Networks}, 1962] and
  \emph{isolating cuts} by [Dahlhaus et al., \emph{SIAM J. Comp.}
  23(4):864-894, 1994].  We sharpen these old weapons with the help of
  submodular functions, and apply them to this problem, which enable
  us to design a more elaborated branching scheme on deciding whether
  a non-terminal vertex is with a terminal or not.  This bounded
  search tree algorithm can be shown to run in $1.84^k\cdot n^{O(1)}$
  time, thereby breaking the $2^k\cdot n^{O(1)}$ barrier.  As a
  by-product, it gives a $1.36^k\cdot n^{O(1)}$ time algorithm for
  $3$-terminal cut.  The preprocessing applied on non-terminal
  vertices might be of use for study of this problem from other
  aspects.
\end{abstract}

\section{Introduction} 

One central and universal problem in combinatorial optimization and
algorithmic graph theory is to find a partition of the vertex set
satisfying some properties.  The most basic formulation is to find a
$2$-partition that separates a source vertex $s$ from a target vertex
$t$, while the crossing edges have minimum cardinality or weight sum,
called the \emph{size} \cite{ford-62-flow}.  Such a partition, called
a minimum ($s,t$)-cut, can be efficiently computed.  This formulation
smoothly extends to the separation of two disjoint subsets of
vertices, both conceptually and computationally.  However, when we
have three or more (sets of) vertices to separate in a pairwise way
with the same minimization objective, it becomes computationally
intractable. 

Dahlhaus et al.~\cite{dahlhaus-94-complexity-of-multiway-cuts} and
Cunningham \cite{cunningham-91-multiway-cut} formulated the
\textsc{multiterminal cut} problem, or \textsc{$p$-terminal cut} when
one wants to emphasize the number $p$ of terminals, and initiated its
study.  The classical but still must-read paper
\cite{dahlhaus-94-complexity-of-multiway-cuts} contains several
results.  First and foremost, it proves its NP-hardness, in general as
well as under special restraints, e.g., $p$ being a constant as small
as $3$, or the input graph being planar.  Another result, also
negative, is that the problem, in its general form or restricted to
constant number of terminals, is MAX SNP-hard.  The most cited
positive result of \cite{dahlhaus-94-complexity-of-multiway-cuts} is
an elegant approximation algorithm with ratio $2(1-1/p)$.  The main
observation behind this algorithm is that an objective $p$-partition
contains for any terminal a $2$-partition that separates it from other
terminals; called an \emph{isolating cut} for this terminal.  On the
one hand, a minimum isolating cut for each terminal can be easily
computed and its size is a lower bound for that of any isolating cut
for this terminal; on the other hand, each crossing edge in the
objective $p$-partition is incident to precisely two parts, and thus
counted twice in the size of all isolating cuts.  Therefore, the size
of any objective cut cannot be smaller than half of the total size of
minimum isolating cuts for all terminals.  Interestingly, the
approximation algorithm from this observation coincides with the
linear program (LP) relaxation from the dual of an integer program for
the multicommodity flow \cite{garg-04-weighted-multiway-cut}, which is
known to have half integrality property.
After a sequence of improvements and with great efforts, Buchbinder et
al.~\cite{buchbinder-13-exponential-clocks-and-mwc} managed to improve
the approximation ratio to $1.323$, which breaks another barrier to
this problem, i.e., $4/3$ for the approximation ratio.

Recently, \textsc{multiterminal cut} and its variations receive a lot
of interest from the perspective of parameterized complexity.  Recall
that a problem, parameterized by $k$, is {\em fixed-parameter
  tractable (FPT)} if there is an algorithm with runtime $f(k)\cdot
n^{O(1)}$, where $f$ is a computable function depending only on $k$.
Using the new concept \emph{important separator}, Marx
\cite{marx-06-parameterized-graph-separation} proved this problem,
parameterized by the cut size $k$, is FPT.  In sequel Chen et
al.~\cite{chen-09-multiway-cut} proposed the first single exponential
time algorithm, with runtime $4^k\cdot n^{O(1)}$.  It was further
improved to $2^k\cdot n^{O(1)}$ by Xiao
\cite{xiao-10-multiterminal-cut}.  Unlike its predecessors, which work
indifferently on \textsc{multiterminal cut} as well as its vertex
variation that is commonly known as \textsc{node multiway cut}, Xiao's
approach, however, applies exclusively to \textsc{multiterminal cut}.
A matching algorithm for \textsc{node multiway cut} was later reported
by Cygan et al.~\cite{cygan-11-multiway-cut}, who used a very
interesting branching scheme based on a novel usage of the LP for
multicommodity flow.  With the benefit of hindsight, we are able to
point out that techniques and results of Xiao
\cite{xiao-10-multiterminal-cut} and Cygan et
al.~\cite{cygan-11-multiway-cut} are essentially the same: the former
algorithm can also be re-interpreted with the same branching scheme
using the same flow-based LP.

The cut size $k$ is not the only parameter for \textsc{multiterminal
  cut}, and a natural alternative is $p$, the number of terminals.
However, as aforementioned, in general graphs it is already NP-hard
when $p=3$, hence very unlikely to be FPT.  On the other hand, when
$p$ is a constant, it can be solved in polynomial time on planar
graphs \cite{dahlhaus-94-complexity-of-multiway-cuts}.  Inspired by
this, it is natural to ask for its fixed-parameter tractability
parameterized by $p$ on planar graphs.  Marx
\cite{marx-12-planar-multiway-cut} recently gave a negative answer to
this question.  This apparently was a part of a systematic study of
this problem, of which other outcomes include a $2^{O(p)}\cdot
n^{O(\sqrt{p})}$ time algorithm \cite{klein-12-planar-multiway-cut},
improving previous one in time $n^{O(p)}$
\cite{dahlhaus-94-complexity-of-multiway-cuts}, and a polynomial time
approximation schema \cite{bateni-12-ptas-planar-multiway-cut}.

Naturally, the hardness of this problem should be attributed to the
interference of isolating cuts for different terminals, i.e., an
isolating cut for one terminal might help separate another pair of
terminals.  For example, consider the following unit-weighted graph: a
path of four vertices whose ends are terminals $t_1$ and $t_2$
respectively, and another terminal $t_3$ adjacent to both internal
vertices of the path.  Clearly, any isolating cut for $t_3$ has size
at least $2$, and there are two cuts with the minimum size.  Of which
one is a proper subset of the other, and the larger one with three
vertices would be more desirable.  This example suggests \emph{a
  containment relation of minimum isolating cuts for a terminal} and
\emph{the preference for the largest one among them}.  They had been
formalized in the classical work of Ford and Fulkerson
\cite{ford-62-flow} and Dahlhaus et al.
\cite{dahlhaus-94-complexity-of-multiway-cuts} (see
Theorems~\ref{thm:max-vol-min-cut} and \ref{thm:min-cut-safe} in
Section~\ref{sec:pre}).\footnote{Interestingly, these theorems seem to
  be widely and unfairly ignored, and hence are re-proved again and
  again even decades after their first appearance.  A plausible
  explanation might be that both pieces of work contain other more
  famous results that eclipse them.}

As a matter of fact, the {important separator} is an immediate
generalization of them.  With this technique, several cut related
problems were shown to be FPT.  Some of them, e.g., \textsc{directed
  feedback vertex set} \cite{chen-08-dfvs} and \textsc{multicut}
\cite{marx-14-multicut,thomasse-10-multicut}, had been open for quite
a long time.  Also related is the concept \emph{extreme set}, which
is very well-known in the study of cuts enumerations
\cite{nagamochi-10-algorithmic-connectivity}.

Here we manage to break the $2^k \cdot n^{O(1)}$ barrier, which
has withstood several rounds of attacks.  
\begin{theorem}\label{thm:main}
  The \textsc{multiterminal cut} problem can be solved in $1.84^k
  \cdot n^{O(1)}$ time.
\end{theorem}
For a fixed number of terminals, we can even do better, especially
when $p$ is a small constant.
\begin{theorem}\label{thm:p-terminal}
  The \textsc{$3$-terminal cut} problem can be solved in $1.36^k \cdot
  n^{O(1)}$ time.  In general, \textsc{$p$-terminal cut} can be solved
  in $1.84^{(p-2)k/(p-1 )} \cdot n^{O(1)}$ time.
\end{theorem}

The main thrust of our algorithm is a careful analysis of vertices
\emph{close} to a terminal in a sense.  Recall that all previous
parameterized algorithms have the same branching scheme, which
considers each undecided vertex and branches on where to place it.
This operation either increases the size of the minimum isolating cut
for a terminal, or identifies a crossing edge; both can be bounded by
functions of $k$.  We show that a set of vertices whose inclusion to a
terminal increases the minimum isolating cut for it by one can be
grouped together, and then disposed of by a more careful branching
rule.  This rule, together with a measure based analysis
\cite{fomin-11-exact-algorithms}, delivers the claimed algorithm.  We would
like to point out its resemblance with \emph{extreme sets}
\cite{nagamochi-10-algorithmic-connectivity}, which might be used in
the preprocessing of non-terminal vertices, and hence shed some light
on the kernelization of \textsc{multiterminal cut}.  Our approach,
however, does not seem to be generalizable to \textsc{node multiway
  cut} in an easy way.  So we leave it open for a parameterized
algorithm for \textsc{node multiway cut} of time $c^k\cdot n^{O(1)}$
for some constant $c < 2$.

A very interesting generalization of the \textsc{multiterminal cut}
problem was recently proposed by Chekuri and Ene
\cite{chekuri-11-submodular-multiway-partition}, who studied the
\textsc{submodular multiway partition} problem, where the cut size is
replaced by general submodular set functions.  The
\textsc{$p$-terminal cut} problem is also related to the
\textsc{$p$-way cut} problem, which asks for a $p$ partition with the
minimum number/weight sum of crossing edges.  Also NP-hard in general,
it can be solved in polynomial time when $p$ is fixed
\cite{goldschmidt-94-k-cut}, and Kawarabayashi and Thorup
\cite{thorup-11-k-cut} recently proved that it is FPT parameterized by
cut size $k$.

\section{Preliminaries}\label{sec:pre}

For notational convenience, we omit braces for singleton sets, i.e.,
we use $t$ to denote both the element itself and the singleton set
that contains only this element.  As the precise meaning is always
clear from context, this abuse will not introduce ambiguities.  We use
$+$ and $-$ for operations set union and set difference,
respectively.  

The graph, given by $(V,E)$, is simple and undirected.  The
cardinality of $V$ is denoted by $n$ throughout the paper.  Each edge
$e\in E$ is weighted by a positive integer $w(e)$; the weight of a
subset $E'$ of $E$ is defined as $w(E') = \sum_{e\in E'} w(e)$.  For
two (possibly intersecting) subsets $X,Y \subset V$, we denote by
$E(X,Y)$ the set of edges with one end in $X$ and the other in $Y$,
and use $w(X,Y)$ as a short hand for $w(E(X,Y))$.  For
a 
subset $X$ of $V$,
the ordered partition $(X, V - {X})$ is called a \emph{cut} of $G$,
and denoted by $X$, whose \emph{size} is defined by $ d(X) = w(X, V -
{X})$.  By definition, $ d(\emptyset) = d(V) = 0$.  A cut $X$ is an
\emph{($S,T$)-cut} if $S\subseteq X \subseteq V - {T}$.  An
{($S,T$)-cut} with the minimum size is called a \emph{minimum
  ($S,T$)-cut}.

The following fact was first observed by Ford and Fulkerson
\cite[Section~1.5]{ford-62-flow}, as a by-product of the max-flow
min-cut theorem, and later rediscovered several times by different
authors.

\begin{theorem}[\cite{ford-62-flow}]\label{thm:max-vol-min-cut}
  Let $S, T \subset V$ be two disjoint nonempty subsets of vertices.
  There is a minimum ($S, T$)-cut $X$ such that all other minimum ($S,
  T$)-cuts are subsets of $X$.  Moreover, this cut can be constructed
  in polynomial time.
\end{theorem}

Such an ($S, T$)-cut will be called the \emph{maximum volume minimum
  ($S,T$)-cut}, or \emph{max-vol min-cut for ($S,T$)} for short.  Note
that the definition is asymmetric and the pair ($S, T$) is ordered;
the sets $S$ and $T$ are commonly referred to as the source and target
terminals, respectively.  Given any pair of terminals, their max-vol
min-cut can be found in polynomial time.  Indeed, most known
algorithms for minimum cut return the max-vol min-cut for either
($S,T$) or ($T,S$).  See also
\cite{nagamochi-10-algorithmic-connectivity} for an updated and more
comprehensive treatment.

Generalized in a natural way, given a set ${ T} = \{t_1,\dots,t_p\}$
of $p$ terminals, where $p\ge 3$, a partition $\{V_1,\dots,V_p\}$ of
$V$ is called a \emph{multiterminal cut} for ${T}$ if $t_i\in V_i$ for
each $1\le i\le p$.\footnote{We define multiterminal cut as a
  partition, instead of a disconnecting set, because it is more
  compatible with the original definition of cuts \cite{ford-62-flow},
  (see Sec.~1.6 of \cite{ford-62-flow} for an explanation on why Ford
  and Fulkerson chose partitions over disconnecting sets in the first
  place,) as well as its extension to submodular multiway partitions
  \cite{chekuri-11-submodular-multiway-partition}.}  The size of this
multiterminal cut is defined to be $\sum_{u\in V_i,v\in V_j, 1\le i<
  j\le p}w(u,v)$, that is, the weight sum of \emph{crossing edges},
edges with ends in different parts.  Since each edge is counted in
precisely two parts, the size is equivalent to
$\frac{1}{2}\sum^p_{i=1} d(V_i)$.  We are now ready for a formal
definition of the problem.

\medskip

\fbox{\parbox{0.9\linewidth}{
  \textsc{multiterminal cut} ($G,w,T,k$)

\begin{tabularx}{\linewidth}{rX}
  \textit{Input:} & a graph $G = (V, E)$, a weight function $w : E \to
  \mathbb{N}$, a set $T$ of $p$ distinct terminals, and a nonnegative
  integer $k$.
  \\
  \textit{Task:} & Either find a multiterminal cut for $T$ of size no
  more than $k$, or report no such a multiterminal cut exists.
\end{tabularx}
}}
\medskip

Needless to say, for each $1\le i\le p$, set $V_i$ makes a ($t_i, T -
t_{i}$)-cut.  In other words, $V_i$ isolates $t_i$ from all other
terminals; hence called an \emph{isolating cut for $t_i$}
\cite{dahlhaus-94-complexity-of-multiway-cuts}.  A \emph{min-iso-cut
  for $t_i$} is a short hand for \emph{minimum isolating cut for
  $t_i$}.  Noting that Theorem~\ref{thm:max-vol-min-cut} applies to
isolating cuts, the \emph{max-vol min-iso-cut for $t_i$} can be
defined naturally.

\begin{theorem}[\cite{dahlhaus-94-complexity-of-multiway-cuts}]
  \label{thm:min-cut-safe}
  Let $(G, w, T, k)$ be an instance of \textsc{multiterminal cut}.
  For any min-iso-cut $X$ for $t_i\in T$, there exists a minimum
  multiterminal cut $\{V_1,\dots,V_p\}$ for $T$ such that $X \subseteq
  V_i$.
\end{theorem}

We say that a set function $f: 2^V \to \mathbb{Z}$ is
\emph{submodular} if every pair of subsets $X,Y\subseteq V$ satisfies
\begin{equation}
  f(X) + f(Y) \ge f(X\cap Y) + f(X\cup Y).
\end{equation}

The following property can be easily verified by definition (see,
e.g., \cite{nagamochi-10-algorithmic-connectivity}), and shall be the
swiss army knife in what follows.
\begin{theorem}\label{thm:submodular}
  The cut size function, $d: 2^V\to \mathbb{Z}$, is submodular.
\end{theorem}

\begin{lemma}\label{lem:increment-isolating-cuts}
  Let ${T} = \{t_1, \dots, t_p\}$ be the set of terminals where $t_i$
  is the max-vol min-cut for ($t_i, T - t_i$) for each $1\le i\le p$.
  If the min-iso-cut size for a terminal decreases with the removal of
  an edge $e=uv$ , then any min-iso-cut for this terminal in $G - e$
  contains either $u$ or $v$, and its size decreases by at most
  $w(e)$.  The upper bound is reached if{f}
  the terminal is $u$ or $v$.  Moreover, the sizes of min-iso-cuts for
  at most two terminals in $G - e$ can be smaller than that in $G$.
\end{lemma}
  \begin{proof}
    Let $C_i$ be any min-iso-cut for $t_i$ in $G - e$.  If neither $u$
    nor $v$ is in $C_i$, then we must have $C_i = t_i$, as otherwise
    $d(C_i) > d(t_i)$, which contradicts the definition of
    min-iso-cut; for the same reason, $u$ and $v$ cannot be both
    contained in $C_i$.  Hence we assume $|\{u,v\}\cap C_i| = 1$.  The
    size of min-iso-cuts for $t_i$ decreases by $d(t_i) - (d(C_i) -
    w(e) )= w(e) + ( d(t_i) - d(C_i) ) \le w(e)$.  To achieve the
    upper bound, we must have $C_i = t_i$, which implies $t_i\in
    \{u,v\}$.

    Suppose to the contrary of the last assertion, both $C_i$ and
    $C_j$, where $i\ne j$, contain $u$.  Then neither of them contains
    $v$.  Let us consider $Y = C_i \cap C_j$.  The minimality of $C_i$
    implies $w(Y, C_i - Y) = w(Y, V - C_i)$; likewise $w(Y, C_j - Y) =
    w(Y, V - C_j)$.  We can conclude that $w(Y, V - (C_i \cup C_j)) =
    0$ and then $d(C_i\setminus C_j) = d(C_i)$.  Therefore,
    $C_i\setminus C_j$ is also a min-iso-cut for $t_i$ in $G - e$; a
    contradiction to the first assertion.  We have seen that $u$ is
    contained in the min-iso-cut for at most one terminal, and the
    same argument applies to $v$.  We can thus conclude the last
    assertion and the lemma.
  \end{proof}
\section{Distances to a terminal}
\label{sec:distances}

For the problem under concern, it will make no difference if a
terminal is replaced by a subset of vertices; they can always be
\emph{merged} into a single vertex.  By merging a subset $U$ of
vertices into a single vertex $u$, we remove all vertices in $U$,
introduce a new vertex $u$, and for each vertex $v\in N(U)$, set the
weight of the new edge $uv$ to be $w(U,v)$.  In particular, if $U$
contains a terminal $t_i$, then the newly introduced vertex will be
also identified as $t_i$; note that we never merge a vertex set
containing more than one terminal.  We remark that this operation does
not introduce non-integral weights.  If a non-terminal vertex $v$ is
decided to be in $V_i$, we can merge $\{v,t_i\}$ into $t_i$.
Theorem~\ref{thm:min-cut-safe} might be informally interpreted as: a
vertex whose merge into a terminal does not increase the min-iso-cuts
size for it can always be safely merged.  This observation inspires us
to consider the increment of the size of min-iso-cuts for $t_i$ with
the merge of a non-terminal vertex $v$ to it.  Noting that for two
disjoint sets $S$ and $T$, and a proper subset $S'$ of $S$, the
max-vol min-cut for ($S',T$) is a subset of that for ($S,T$), thereby
having a size no larger than the latter.

\begin{definition}
  Let $v$ be a non-terminal vertex and $t_i$ be a terminal.  Let $C'$
  and $C$ denote the max-vol min-cuts for ($t_i + v, T - t_i$) and,
  respectively, ($t_i, T - t_i$).  The \emph{extension} of $v$ from
  $t_i$ is defined to be $X_i(v) = C' - C$, and the \emph{distance} of
  $v$ from $t_i$ is defined to be $ ext_i(v) = d(C') - d(C)$.
\end{definition}
\\*
Using Theorem~\ref{thm:max-vol-min-cut}, we can compute the extension
and distance for any pair of non-terminal vertex $v$ and terminal
$t_i$ in polynomial time.  By definition, an extension contains no
terminal, and a distance is always nonnegative.  In particular, the
extension is empty if and only if the distance is $0$, then $v\in C$.
Although the set $C'$ might not be connected (even when $C$ is
connected), the extension always is:

\begin{lemma}
  \label{lem:X-connected}
  For any non-terminal vertex $v$ and terminal $t_i$ , the subgraph
  induced by $X_i(v)$ is connected.
\end{lemma}
  \begin{proof}
    We may assume $X_i(v)$ is nonempty.  Suppose, for contradiction,
    that $X_i(v)$ has a partition $\{Y,Z\}$ such that there is no edge
    between $Y$ and $Z$, and without loss of generality, $v \in Y$.
    Let $C$ be the max-vol min-iso-cut for $t_i$.  Since $Y$ and $Z$
    are disconnected and disjoint from $C$, we have $ d(C + Z) >
    d(C)$, and
    \[
     d(C + X_i(v)) -  d(C) =  d(C + Y) -  d(C)
    +  d(C + Z) -  d(C) >  d(C + Y) -  d(C),
    \]
    which means $C + Y$ is a strictly smaller cut for ($t_i + v, T -
    t_i$) than $C + X_i(v) = C + Y + Z$; a contradiction.  
  \end{proof}

  Let us now put those vertices in the immediate vicinity to terminal
  $t_i$, i.e., vertices whose distances to $t_i$ are exactly $1$,
  under a closer scrutiny.

\begin{lemma}
  \label{lem:min+1}
  Let $v$ be a vertex at distance $1$ to terminal $t_i$.  There exists
  a minimum {multiterminal cut} for $T$ that keeps $X_i(v)$ in the
  same part.
\end{lemma}
  
\begin{proof}
  We denote by $C$ and $C'$ the max-vol min-cuts for ($t_i, T - t_i$)
  and ($t_i + v, T - t_i$), respectively, and let $X = X_i(v)$.  We
  prove this lemma by constructing a claimed cut.  Let ${\cal P} =
  \{V_1, V_2, \dots, V_p\}$ be a minimum multiterminal cut for $T$
  satisfying $C\subseteq V_i$; its existence is ensured by
  Theorem~\ref{thm:min-cut-safe}.  If $X\subset V_j$ for some $1\le
  j\le p$, then we are already done; hence we assume otherwise.  We
  claim that ${\cal P}' = \{V'_1, V'_2, \dots, V'_t\}$, where $V'_i =
  V_i \cup X$ and $V'_j = V_j \setminus X$ for $j\ne i$, is also a
  minimum multiterminal cut for $T$.  The remaining discussion is
  separated based on $X \cap V_i = \emptyset$ or not.

  If $X \cap V_i \neq \emptyset$, then $C$ is proper subset of $V_i
  \cap C'$, which implies $d(C) < d(V_i \cap C')$.  Noting all values
  are integral, we have $d(C') = d(C) + 1 \le d(V_i \cap C')$.  By the
  submodularity of the cut function (Theorem~\ref{thm:submodular}), we
  have $d(V_i) + d(C') \geq d(V_i \cup C') + d(V_i \cap C')$.
  Combining them, we get
    \[
      d(V_i) \geq  d(V_i \cup C') = d(V_i \cup X) =  d(V'_i).
    \]
    Noting that $E(V'_j,V'_l)\subseteq E(V_j,V_l)$ for any pair of
    $j,l$ different from $i$, this case is proved.

    Otherwise, $X \cap V_i = \emptyset$.  The crossing edges of ${\cal
      P}'$ but not of ${\cal P}$ are a subset of $E(X, V - {C'})$;
    while the crossing edges of ${\cal P}$ but not of ${\cal P}'$
    include $E(V_i,X)$ and at least one edge $e$ in $G[X]$: noting
    that by Lemma~\ref{lem:X-connected}, $G[X]$ was originally
    connected. Therefore the difference of the cut sizes of ${\cal
      P}'$ and ${\cal P}$ is at most:
    \[
    w(X, V - {C'}) - w(V_i, X) - w(e) \leq ext_i(u) - 1 = 0,
    \]
    which justifies the optimality of ${\cal P}'$ and finishes the
    proof.
  
\end{proof}

Lemma~\ref{lem:min+1} permits us to merge each extension $X_i(v)$ with
$ext_i(v) = 1$ into a single non-terminal vertex.  After both the
max-vol min-iso-cut for $t_i$ and $X_i(v)$ have been merged into $t_i$
and $v$ respectively, either they are adjacent, or the new vertex $v$
is only incident to a unit-weight edge.  As vertices of degree $1$ can
always be removed safely, we may assume that all vertices at unit
distances to $t_i$ are its neighbors.

\section{The algorithm}

From Theorems \ref{thm:max-vol-min-cut} and \ref{thm:min-cut-safe},
one can easily derive a bounded search tree algorithm as follows.  For
$1\le i\le p$, we initialize $V_i$ to be the max-vol min-iso-cut for
$t_i$ and merge it into $t_i$.  Then we grow it by including its
neighbors one by one, until no vertex is left out.  On each neighbor
$v$ of $t_i$, we have two options: including it and merging
$\{v,t_i\}$ into $t_i$, or excluding it and then edge $(t_i,v)$ is a
crossing edge.
This branching rule was first presented by Chen et
al.~\cite{chen-09-multiway-cut} and later used by Xiao
\cite{xiao-10-multiterminal-cut}.  The original analysis in
\cite{chen-09-multiway-cut} uses $2k - d(t_i)$ as the measure, which
delivers a $4^k$ bound on the number of leaves the algorithm traverses
in the search tree.  Inspired by the observation on isolating cuts,
Xiao \cite{xiao-10-multiterminal-cut} used the new measure $2k - \sum
d(t_i)$.  As it always holds that $\sum d(t_i)>k$ in a nontrivial
instance, the bound can be tightened to $2^k$.  This analysis is then
tight: there are instances of \textsc{multiterminal cut} on which this
branching rule will make a search tree of $2^k$ leaves.  In other
words, the base $2$ cannot be further lowered by analysis.

Let $m$ be the measure $2k - \sum d(t_i)$.  The bound $2^k$ follows
from $m\le k$ and the branching vector ($1, 1$): it forks into at most
two branches, in each of which the measure decreases by at least
$1$.\footnote{We shall use standard technique to analyze the bounded
  search tree algorithm (see, e.g.,
  \cite[Theorem~2.1]{fomin-11-exact-algorithms}).  We say that a branching rule
  has \emph{branch vector} ($\tau_1, \tau_2,\dots, \tau_r$) if, given
  an input instance of measure $m$, it branches into $r$ instances,
  and the measures in them are at most $m-\tau_1, m-\tau_2, \dots,
  m-\tau_r$, respectively.  With such a branch vector, the linear
  recurrence for the maximum number of leaves is
\[
{\cal T}(m) \le {\cal T}(m - \tau_1) + {\cal T}(m - \tau_2) + \dots
+{\cal T}(m - \tau_r),
\]
and then ${\cal T}(m) \le c^m$, where $c$ is the unique positive real
root of
\[
z^m - z^{m - \tau_1} - z^{m - \tau_2} - \dots - z^{m - \tau_r} = 0.
\]} To break the $2^k$ barrier, a branching vector strictly better
than ($1, 1$) is in need.  Observing that the inclusion of a vertex
$v$ always increases $\sum d(t_i)$ by the distance of $v$ to the
terminal, vertices at distance $2$ or more will not concern us; hence
we only need to concentrate on those vertices at unit distance to a
terminal.  On these vertices we design new branching rules, worst of
which has branching vector ($1, 2, 3$); thereby achieving the claimed
bound of Theorem~\ref{thm:main}.

Let $v$ be a non-terminal vertex.  If it has a unique neighbor $u$,
then $v$ and $u$ will always be together; if it has two neighbors $u$
and $u'$, and without loss of generality, $w(v,u) \ge w(v,u')$, then
there always exists a minimum multiterminal cut for $T$ that keeps $u$
and $v$ together.  This holds regardless whether $u$ and/or $u'$ is a
terminal or not.  
In both cases, we may merge $v$ to $u$.  Now we are ready to present
our algorithm in Figure~\ref{fig:alg}.  Every time when the terminal
$t_p$ has been separated from the rest of the graph, i.e., the part
$V_p$ has been identified, we treat the rest of the graph as a new
instance with one less terminal.  Nonessential particulars on the
bookkeeping and recovering of merge operations is omitted for the
simplicity.  When the value of $h = \sum d(t_i)$ changes, it is marked
at the end of the line that makes it happen; $c^+$ means at least $c$.

\begin{figure*}[h]
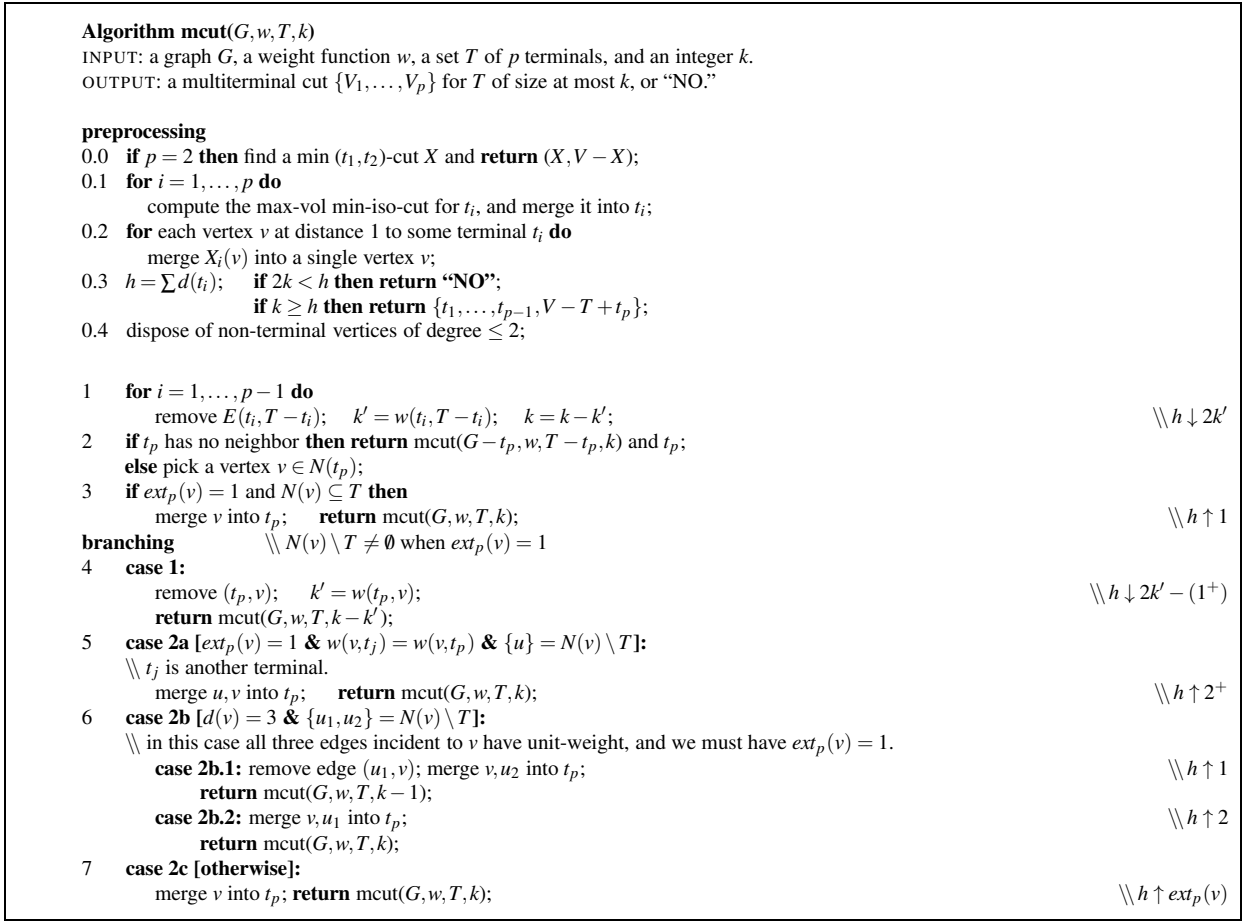

\setbox4=\vbox{\hsize28pc \noindent\strut
\begin{quote}
  \vspace*{-5mm} \footnotesize
  {\bf Algorithm mcut($G, w, {T},k$)}\\
  {\sc input}: a graph $G$, a weight function $w$, a set $T$ of $p$
  terminals, and an integer $k$.
  \\
  {\sc output}: a multiterminal cut $\{V_1,\dots,V_p\}$ for $T$ of
  size at most $k$, or ``NO.''
  \\
  \\
  {\bf preprocessing}
  \\
  0.0 \hspace*{1mm} {\bf if} $p = 2$ {\bf then} find a min
  ($t_1,t_2$)-cut $X$ and {\bf return} ($X,V-X$);
  \\
  0.1 \hspace*{1mm} {\bf for} $ i = 1, \dots, p$ {\bf do}
  \\
  \hspace*{8mm} compute the max-vol min-iso-cut for $t_i$, and merge
  it into $t_i$;
  \\
  0.2 \hspace*{1mm} {\bf for} each vertex $v$ at distance $1$ to some
  terminal $t_i$ {\bf do}
  \\
  \hspace*{8mm} merge $X_i(v)$ into a single vertex $v$;
  \\
  0.3 \hspace*{1mm} $h = \sum d(t_i)$;
  $\quad$ \parbox[t]{0.85\linewidth}{ {\bf if} $2k < h$ {\bf then
      return ``NO''};
    \\
    {\bf if} $k \ge h$ {\bf then return} $\{t_1, \dots, t_{p-1}, V -
    {T} + t_p \}$;}
  \\
  0.4 \hspace*{1mm} dispose of non-terminal vertices of degree $\le
  2$;
  \\

  1 \hspace*{3mm} {\bf for} $ i = 1, \dots, p-1$ {\bf do}
  \\
  \hspace*{9mm} remove $E(t_i, T - t_i)$; $\quad k' = w(t_i, T -
  t_i)$; $\quad k=k- k'$; \hfill $\setminus\!\!\!\setminus h\downarrow
  2k'$
  \\
  2 \hspace*{3mm} {\bf if} $t_p$ has no neighbor {\bf then return}
  mcut($G - t_p, w, {T - t_p}, k$) and $t_p$;
  \\
  \hspace*{5mm} {\bf else} pick a vertex $v\in N(t_p)$;
  \\
  3 \hspace*{3mm} {\bf if} $ext_p(v) = 1$ and $N(v)\subseteq T$ {\bf
    then}
  \\
  \hspace*{9mm} merge $v$ into $t_p$; $\quad$ {\bf return} mcut($G, w,
  {T}, k$); \hfill $\setminus\!\!\!\setminus h\uparrow 1$
  \\*
  {\bf branching} $\qquad\qquad
  \setminus\!\!\setminus$ $N(v)\setminus T\ne\emptyset$ when $ext_p(v)
  = 1$
  \\
  4 \hspace*{3mm} {\bf case 1:}
  \\
  \hspace*{9mm} remove $(t_p, v)$; $\quad$ $k' = w(t_p, v)$; \hfill
  $\setminus\!\!\!\setminus h\downarrow 2k'-(1^+)$
  \\
  \hspace*{9mm} {\bf return} mcut($G, w, {T}, k - k'$);
  \\
  5 \hspace*{3mm} {\bf case 2a [$ext_p(v) = 1$ \& $w(v, t_j) = w(v,
    t_p)$ \& $\{u\} = N(v) \setminus T$]:}
  \\
  \hspace*{5mm} $\setminus\!\!\setminus$ $t_j$ is another terminal.
  \\
  \hspace*{9mm} merge $u,v$ into $t_p$; $\quad$ {\bf return} mcut($G,
  w, {T}, k$); \hfill $\setminus\!\!\!\setminus h\uparrow 2^+$
  \\
  6 \hspace*{3mm} {\bf case 2b [$ d(v)=3$ \& $\{u_1,u_2\} = N(v)
    \setminus T$]:}
  \\
  \hspace*{5mm} $\setminus\!\!\setminus$ in this case all three edges
  incident to $v$ have unit-weight, and we must have $ext_p(v)=1$.
  \\
  \hspace*{9mm} {\bf case 2b.1:} remove edge $(u_1, v)$; merge $v,u_2$
  into $t_p$; \hfill $\setminus\!\!\!\setminus h\uparrow 1$
  \\
  \hspace*{15mm} {\bf return} mcut($G, w, {T}, k - 1$);
  \\
  \hspace*{9mm} {\bf case 2b.2:} merge $v,u_1$ into $t_p$; \hfill
  $\setminus\!\!\!\setminus h\uparrow 2$
  \\
  \hspace*{15mm} {\bf return} mcut($G, w, {T}, k$);
  \\
  7 \hspace*{3mm} {\bf case 2c [otherwise]:}
  \\
  \hspace*{9mm} merge $v$ into $t_p$; {\bf return} mcut($G, w, {T},
  k$); \hfill $\setminus\!\!\!\setminus h\uparrow ext_p(v)$

\end{quote} \vspace*{-6mm} \strut} $$\boxit{\box4}$$
\vspace*{-9mm}
\caption{Algorithm for \textsc{multiterminal cut}}
\label{fig:alg}
\end{figure*}

In one case of our analysis, we will need the following lemma.
\begin{lemma}\label{lem:case1}
  In case 1 (step 4) of algorithm {\bf mcut}, $h$ decreases by at
  most $2w(t_p, v) - 1$.  Moreover, when $w(t_p, v) > 1$ and $ext_p(v)
  = 1$, the value of $h$ decreases by $2 w(t_p, v) - 1$ if and only if
  there is another terminal $t_j$ such that $w(t_p, v) = w(t_j, v)$.
\end{lemma}
\begin{proof}
  As all edges between $t_p$ and other terminals have been removed in
  step 1, the first assertion follows immediately from
  Lemma~\ref{lem:increment-isolating-cuts}.  The ``if'' direction of
  the second assertion is clear, and we now prove the ``only if''
  direction.

  Assume $t_j$ is the other affected terminal, then the size of
  min-iso-cuts for $t_j$ decreases by $w(t_p, v) - ext_j(v)$; thus
  $ext_j(v) = 1$, and $v$ and $t_j$ are adjacent.  Noting that
  $ext_p(v)$ is defined to be
    \begin{align*}
      & w(v, V - t_p) - w(v, t_p)\\
      =& w(v, t_j) + w(v, V - t_p - t_j) - w(v, t_p)\\
      =& w(v, V - t_j) - 1 + w(v, V - t_p - t_j) - w(v, t_p)\\
      =& w(v, V - t_p - t_j) - 1 + w(v, V - t_p - t_j),
    \end{align*}
    we can conclude $w(v, V - t_p - t_j) = 1$, and $w(t_p, v) = w(t_j,
    v) = \frac{ d(v) - 1}{2}$.
    
\end{proof}

\begin{theorem}
  Algorithm {\bf mcut} solves \textsc{multiterminal cut} in time
  $1.84^k \cdot n^{O(1)}$.
\end{theorem}
\begin{proof}
  Let us first verify its correctness.  Step 0.0 ensures that the
  instance has at least three terminals; otherwise it solves the
  instance by finding a minimum ($t_1, t_2$)-cut.  The correctness of
  preprocessing steps~0.1 and 0.2 follows immediately from
  Theorem~\ref{thm:min-cut-safe} and Lemma~\ref{lem:min+1},
  respectively.  After that, on any given multiterminal cut $\{V_1,
  \dots, V_p\}$ for $T$, we have $d(V_i)\ge d(t_i)$ for each $1\le
  i\le p$; in other words, this cut has size $1/2\sum^p_{i=1} d(V_i)
  \ge 1/2\sum^p_{i=1} d(t_i)= h/2$.  Therefore, if $h > 2k$ we can
  safely report that there is no solution of size $k$ for the input
  instance.  On the other hand, the multiterminal cut $\{t_1, \dots,
  t_{p-1}, V - {T} + t_p \}$ has size no more than $h$, and hence can
  be returned as a solution if $h \leq k$.  This justifies step 0.3.
  Step 0.4 is straightforward.  A posteriori, after the preprocessing,
  the max-vol min-iso-cut for each terminal consists of itself; and
  any non-terminal vertex has at least three neighbors.

  Steps 1 and 2 are clear: every edge connecting two terminals is a
  crossing edge; an isolated terminal can be disregarded for further
  consideration.  Step 3 takes care of the case that $v$ has
  unit-distance to $t_p$ and is only adjacent to terminals.  Since we
  have disposed of non-terminal vertices of degree $1$ or $2$ in step
  0.4, and step 1 deletes only edges between terminals, there are at
  least two other terminals adjacent to $v$.  Noting $ext_p(v) = 1$
  and all weights are integral, it follows that $w(v,t_p)$ is the
  largest among all edges incident to $v$, and thus putting $v$ into
  $V_p$ will minimize crossing edges incident to $v$.  As $v$ is
  nonadjacent to any other non-terminals, its choice has no effect on
  other vertices.  This justifies step 3.

  Hereafter it takes the branching steps, and enters exclusively one
  case of them, which then calls the algorithm recursively.  For each
  neighbor $v$ of $t_p$, we have to either include $v$ to $V_p$ or
  count $(v,t_p)$ as a crossing edge.  Step 4 deals with the latter
  option, where we remove the edge and decrease the parameter
  accordingly.  The other option, where $v$ is put into $V_p$, is
  handled by steps 5-7.
  
  Step 5 handles the case where $v$ is balanced in two terminals $t_p$
  and $t_j$, then there must be another non-terminal vertex $u$ such
  that $w(v,u) = 1$.  As $w(v, t_p) = w(v, t_j)\ge 1 = w(v, u)$, we
  may assume $v$ is in either $V_p$ or $V_j$.  If $u$ is not in the
  same part as $v$, then it does not matter where to put $v$.  We may
  put it in $V_j$ and thus count $(v,t_p)$ as a crossing edge, which
  is already covered by case 1 (step 4).  Therefore, we may assume
  that $v$ is in $V_p$ only if $u$ is also in $V_p$.

  Step 6 handles the case where $v$ is adjacent to three neighbors
  each with a unit-weight edge.  If an optimal partition puts the
  three neighbors of $v$ into different parts, then $v$ must be in one
  of them and it does not matter which one.  This situation has
  already been covered in case 1.  Hence we assume either $N(v)$ is in
  the same part, which is covered by case 2b.2; or separated into two
  parts, which is covered by case 1, case 2b.1, and case 2b.2
  respectively.

  The correctness of step 7 is straightforward.  This completes the
  proof of the correctness of the algorithm.

  It is clear that every step takes polynomial time.  Hence we will
  concentrate on the branching steps, where it makes 2 or 3 recursive
  calls to itself.  In particular, it either executes case 1, or one
  of the cases 2a, 2b, and 2c.  To precisely determine the time
  complexity, we define $m = 2k - h$ as the measure, and consider its
  decrease in each branch.  By assumption that $t_p$ is the max-vol
  min-iso-cut for $t_p$, and $X_p(v) = v$ when $ext_p(v) = 1$, for
  every other non-terminal vertex $u$ (different from $v$), we always
  have
  \begin{equation}\label{eq:2}
    d(\{t_p, v, u\})\ge d(t_p) + 2.
  \end{equation}
  Also note that the measure will not increase as a side effect in
  other steps.

  In case 1, $k$ decreases by $w(t_p,v)$; by Lemma~\ref{lem:case1},
  $h$ decreases by at most $2w(t_p,v) - 1$.  In total, $m$ decreases
  by at least $1$.
  \begin{itemize}
  \item Case 2a.  According to \eqref{eq:2}, $h$ increases by at least
    $2$; no edge is deleted and $k$ is unchanged.  Thus, $m$ decreases
    by at least $2$.
  \item Case 2b is further separated into two sub-cases.  
    \begin{itemize}
    \item Case 2b.1.  According to \eqref{eq:2}, the min-iso-cut size
      for $t_p$ increases by 1.  By
      Lemma~\ref{lem:increment-isolating-cuts} and noting that the
      only deleted edge $u_1 v$ has unit-weight and $u_1$ is not a
      terminal, the min-iso-cut for other terminals are unchanged.
      Thus, $h$ increases by at least $1$.  As $k$ decreases by $1$,
      the measure decreases by at least $3$.
    \item Case 2b.2.  According to \eqref{eq:2}, $h$ increases by at
      least $2$ while $k$ remains unchanged.  Thus, $m$ decreases by at
      least $2$.
    \end{itemize}
  \item Case 2c.  The value of $h$ increases by $ext_p(v)$ while $k$
    remains unchanged; hence $m$ decreases by $ext_p(v)$.  According
    to Lemma~\ref{lem:case1}, $h$ decreases by $1$ in case 1 only if
    the condition of case 2b is satisfied.  In other words, if $m$
    decreases $1$ in case 2c, then it decreases by at least $2$ in
    case 1.
  \end{itemize}

  In summary, the worst case is when the algorithm branches into case
  1 and 2b.  In this situation, this algorithm makes three recursive
  calls to itself, with new measures $m-1$, $m-2$, and $m-3$,
  respectively.  Recall that we only start the branching when $h > k$,
  which means the initial value of the measure is $m = 2k - h < k$.
  We terminate the branching by returning ``NO'' as long as $m$
  decreases to a negative value, which means that the measure
  decreases by at most $k$, and therefore the number of recursive
  calls is upper bounded by $1.84^k$.  This finishes the analysis of
  time complexity.
\end{proof}

This theorem implies Theorem~\ref{thm:main}.  

Given a set of isolating cuts $(C_1,\dots,C_p)$, taking any $p - 1$ of
them while leaving all remained vertices in one part will make a
multiterminal cut.  If we omit the one with the largest size, then the
cut obtained as such will have size at most $(1-1/p)\sum^p_{i=1}
d(C_i)$.  This observation was first used to improve the approximation
ratio from $2$ to $2(1-1/p)$
\cite{dahlhaus-94-complexity-of-multiway-cuts}, and is also applicable
for our analysis.  The improvement is especially significant when $p$
is small.

\begin{proof}[Proof of Theorem~\ref{thm:p-terminal}]
  The exit condition $k \geq h$ in step 0.3 of algorithm \textbf{mcut}
  can be replaced by $k \geq (1-1/p) h$.  In other words, we only
  branch when $\frac{p}{p-1}k < h \leq 2k$, which means the measure $m
  = 2k - h$ can decrease by at most
  \[
  2k - \frac{p}{p-1}k = \frac{p-2}{p-1}k,
  \]
  from which and the branching vector
  ($1,2,3$) Theorem~\ref{thm:p-terminal} follows.  
  
\end{proof}

\end{document}